\documentclass[runningheads,a4paper,envcountsame]{llncs}

\usepackage{graphicx}
\usepackage[utf8]{inputenc}
\usepackage[T1]{fontenc}
\usepackage{color}
\usepackage{amsmath}
\usepackage{amsfonts, amssymb, dsfont}
\usepackage{todonotes}
\usepackage{algorithm}
\usepackage{algpseudocode}
\usepackage{float}
\usepackage{wasysym}
\usepackage{blkarray,booktabs,bigstrut}
\usepackage{verbatim}
\usepackage{tabularx}
\usepackage{environ}
\usepackage{hyperref}
\usepackage{cleveref}
\usepackage{upgreek}
\usepackage{thmtools, thm-restate}
\usepackage[nocompress]{cite}

\DeclareMathOperator{\Fact}{Fact}

\DeclareMathOperator{\Pref}{Pref}

\DeclareMathOperator{\Suff}{Suff}

\DeclareMathOperator{\var}{var}

\def\N{\mathbb{N}}
\def\ta{\mathtt{a}}
\def\tb{\mathtt{b}}

\newcommand{\al}{\operatorname{alph}}

\begin{document}

\title{The Equivalence Problem of E-Pattern Languages with Regular Constraints is Undecidable}

\titlerunning{On Pattern Languages with Regular Constraints}

\author{Dirk Nowotka, Max Wiedenhöft\thanks{This work was supported by the DFG project number 437493335.}}

\authorrunning{D. Nowotka, M. Wiedenhöft}

\institute{{Department of Computer Science, Kiel University, Germany}
\email{\{dn,maw\}@informatik.uni-kiel.de}}

\maketitle             

\begin{abstract}
Patterns are words with terminals and variables. The language of a pattern is
the set of words obtained by uniformly substituting all variables with words that contain only terminals. 
Regular constraints restrict valid substitutions of variables by
associating with each variable a regular language representable by, e.g., finite automata. 
Pattern languages with regular constraints contain only words in which each variable is substituted 
according to a set of regular constraints.
We consider the membership, inclusion, and equivalence problems for erasing
and non-erasing pattern languages with regular constraints.
Our main result shows that the erasing equivalence problem—one of the most prominent open problems 
in the realm of patterns—becomes undecidable if regular constraints are allowed in addition to variable equality.

\end{abstract}

\section{Introduction}
A \emph{pattern} is a finite word consisting of symbols from a finite set of letters $\Sigma = \{a_1, ... , a_\sigma\}$, also called terminals, and
from an infinite set of variables $X = \{x_1, x_2, ...\}$ with $\Sigma \cap X = \emptyset$. It is a natural and compact device to define formal languages.
Words consisting of only terminal symbols are obtained from patterns by a \emph{substitution} $h$, a terminal preserving morphism which maps all variables 
from a pattern to words over the terminal alphabet. The \emph{language} of a pattern consists of all words obtainable from that pattern by substitutions.

\vspace{2mm}

We differentiate between two kinds of substitutions. Originally, pattern languages introduced by Angluin \cite{DBLP:journals/jcss/Angluin80}
only consisted of words obtained by \emph{non-erasing substitutions} that required all variables to be mapped to non-empty words. Thus, those
languages are also called \emph{NE-pattern languages}. Later, so called \emph{erasing-}/\emph{extended-} or just \emph{E-pattern languages} have been 
introduced by Shinohara \cite{Shinohara1983}. In these, substitutions are also allowed to map variables to the empty word. 
Consider, for example, the pattern $\alpha := x_1\ta\tb x_2x_2$. Then, by mapping $x_1$ to $\ta\ta\ta$ and $x_2$ to $\tb\ta$ with a substitution $h$,
we obtain the word $h(\alpha) = \ta\ta\ta\ta\tb\tb\ta\tb\ta$. If we consider the E-pattern language of $\alpha$, we could also map $x_2$ to the empty
word $\varepsilon$ with a substitution $h'$ which also maps $x_1$ to $\ta\ta\ta$ and obtain $h'(\alpha) = \ta\ta\ta\ta\tb$.

\newpage

Due to its practical and simple definition, patterns and their corresponding languages occur in numerous areas regarding computer science and 
discrete mathematics, including unavoidable patterns \cite{Jiang1994,lothaire1997}, algorithmic learning theory \cite{DBLP:journals/jcss/Angluin80,FERNAU201844,Shinohara1995},
word equations \cite{lothaire1997}, theory of extended regular expressions with back references \cite{FREYDENBERGER20191},
and database theory \cite{FreydenbergerP21,SchmidSchweikardtPODS2022}.

The main problems regarding patterns and pattern languages are the \emph{membership problem} (and its variations \cite{gawrychowski_et_al:LIPIcs.MFCS.2021.48,Manea2022, Fleischmann2023}), the \emph{inclusion problem}, and the \emph{equivalence problem} in both the erasing (E) and non-erasing (NE) cases.
The membership problem determines if a word belongs to a pattern's language. This problem is NP-complete for both E- and NE-pattern languages \cite{DBLP:journals/jcss/Angluin80,Jiang1994}.
The inclusion problem asks if one pattern's language is included in another's. Jiang et al. \cite{Jiang1995} showed that it is generally undecidable for E- and NE-pattern languages. 
Freydenberger and Reidenbach \cite{Freydenberger2010}, and Bremer and Freydenberger \cite{BREMER201215} proved its undecidability for all alphabets with size $\geq 2$ in both E- and NE-pattern languages.
The equivalence problem tests if two patterns generate the same language. 
It is trivially decidable for NE-pattern languages \cite{DBLP:journals/jcss/Angluin80}.
Whether its decidable for E-pattern languages is one of the major open problems in the field \cite{Jiang1995,Reidenbach2004-1,Reidenbach2004-2,Ohlebusch1996,Reidenbach2007}. 
However, for terminal-free patterns, the inclusion and equivalence problems in E-pattern languages have been characterized and shown to be NP-complete \cite{Jiang1995,DBLP:journals/ipl/EhrenfeuchtR79a}. 
The decidability of the inclusion problem for terminal-free NE-pattern languages remains unresolved, though.

Various extensions to patterns and pattern languages have been introduced over time. Some examples are the  
bounded scope coincidence degree, patterns with bounded treewidth, $k$-local patterns, and strongly-nested patterns (see \cite{Day2018} and references therein). 
Koshiba \cite{Koshiba1995} introduced
so called \emph{typed patterns} to enhance the expressiveness of pattern languages by restricting substitutions of variables to types,
i.e., arbitrary recursive languages. This has recently been extended by Geilke and Zilles \cite{Geilke2011} who introduced the notion of \emph{relational patterns} and \emph{relational pattern languages}.

We consider a specific class of typed- or relational patterns called \emph{patterns with regular constraints}. 
Let $\mathcal{L}_{Reg}$ be the set of all regular languages.
Then, we say that a mapping $r : X \rightarrow \mathcal{L}_{Reg}$ is a \emph{regular constraint} that implicitly defines \emph{languages on variables}
$x\in X$ by $L_{r}(x) = r(x)$. Let $\mathcal{C}_{Reg}$ be the \emph{set of all regular constraints}.
A \emph{patterns with regular constraints} $(\alpha,r_\alpha)\in (\Sigma\cup X)^* \times \mathcal{C}_{Reg}$ is a pattern which is associated with
a regular constraint. A substitution $h$ is \emph{$r_\alpha$-valid} if all variables are substituted according to $r_\alpha$.
The language of $(\alpha,r_{\alpha})$ is defined analogously to pattern languages with the additional requirement that all substitutions
must be $r_\alpha$-valid.

This paper examines erasing (E) and non-erasing (NE) pattern languages with regular constraints. 
The membership problem for both is NP-complete, 
while the inclusion problem is undecidable for the general and terminal-free versions.
This immediately follows from known results.
The main finding of this paper is that the equivalence problem for erasing pattern 
languages with regular constraints is indeed undecidable.

\section{Preliminaries}
Let $\N$ denote the natural numbers $\{1, 2, 3, \dots\}$
and let $\N_0 := \N \cup \{0\}$.
For $n,m \in \mathbb{N}$ set $[m,n] := \{k \in \mathbb{N} \mid m \leq k \leq n\}$. 
Denote $[n] := [1,n]$ and $[n]_0 := [0,n]$.
The powerset of any set $A$ is denoted by $\mathcal{P}(A)$. 
An \emph{alphabet} $\Sigma$ is a non-empty finite set whose elements are called \emph{letters}.  
A \emph{word} is a finite sequence of letters from $\Sigma$. 
Let $\Sigma^*$ be the set of all finite words over $\Sigma$, thus it is a free monoid with concatenation as operation and the empty 
word $\varepsilon$ as natural element. Set $\Sigma^+ := \Sigma^* \setminus \{\varepsilon\}$.
We call the number of letters in a word $w \in \Sigma^*$ \emph{length} of $w$, denoted by $|w|$.
Therefore, we have $|\varepsilon| = 0$.
If $w = xyz$ for some $x,y,z\in\Sigma^*$, we call $x$ a \emph{prefix} of $w$, $y$ a \emph{factor} of $w$,
and $z$ a \emph{suffix} of $w$ and denote the sets of all prefixes, factors, and suffixes of $w$ by $\Pref(w)$, $\Fact(w)$,
and $\Suff(w)$ respectively.
For words $w,u\in\Sigma^*$, let $|w|_u$ denote the number of distinct occurrences of $u$ in $w$ as a factor.
Denote $\Sigma^k := \{w \in \Sigma^* \mid |w| = k\}$.
For $w \in \Sigma^*$, let $w[i]$ denote $w$'s $i^{th}$ letter for all $i \in [\vert w \vert]$. 
For reasons of compactness, we denote $w[i] \cdots w[j]$ by $w[i \cdots j]$ for all $i,j \in [\vert w \vert]$ with $i < j$.
Set $\al(w) :=  \{\ta \in \Sigma \mid \exists i \in [\vert w \vert] : w[i] = \ta\}$ as $w$’s alphabet.

Let $X$ be a countable set of variables such that $\Sigma \cap X = \emptyset$.
A \emph{pattern} is then a non-empty, finite word over $\Sigma \cup X$.
The set of all patterns over $\Sigma \cup X$ is denoted by $Pat_\Sigma$. 
For example, $x_1 \ta x_2 \tb \ta x_2 x_3$ is a pattern over $\Sigma = \{\ta,\tb\}$ with $x_1,x_2,x_3\in X$.
For a pattern $p\in Pat_\Sigma$, let $\var(p) := \{\ x \in X\ |\ |p|_x \geq 1\ \}$ denote the set of variables occurring in $p$.
A \emph{substitution of $p$} is a morphism $h : (\Sigma \cup X)^* \to \Sigma^*$ such that $h(\ta) = \ta$ for all $\ta \in \Sigma$ and 
$h(x) \in \Sigma^*$ for all $x \in X$. 
If we have $h(x) \neq \varepsilon$ for all $x \in \var(p)$, we call $h$ a \emph{non-erasing substitution} for $p$. 
Otherwise $h$ is an \emph{erasing substitution} for $p$. The set of all substitutions w.r.t.~$\Sigma$ is denoted by $H_\Sigma$.
If $\Sigma$ is clear from the context, we may write just $H$.
Given a pattern $\alpha\in Pat_\Sigma$, it's erasing pattern language $L_E(\alpha)$ and its non-erasing pattern language $L_{NE}(\alpha)$
are defined respectively by
\begin{align*}
	L_{E}(\alpha) &:= \{\ h(\alpha)\ |\ h\in H, h(x) \in \Sigma^* \text{ for all } x\in\var(\alpha)\}, \text{ and } \\
	L_{NE}(\alpha) &:= \{\ h(\alpha)\ |\ h\in H, h(x) \in \Sigma^+ \text{ for all } x\in\var(\alpha)\}.
\end{align*}

Let $\mathcal{L}_{Reg}$ be the set of all regular languages. We call a mapping $r : X \rightarrow \mathcal{L}_{Reg}$
a \emph{regular constraint} on $X$. If not stated otherwise, we always have $r(x) = \Sigma^*$.
We denote the \emph{set of all regular constraints} by $\mathcal{C}_{Reg}$.
For some $r\in\mathcal{C}_{Reg}$ we define the \emph{language of a variable} $x\in X$ by $L_r(x) = r(x)$.
If $r$ is clear by the context, we omit it and just write $L(x)$.
A \emph{pattern with regular constraints} is a pair $(p,r_p) \in Pat_\Sigma\times\mathcal{C}_{Reg}$.
We denote the \emph{set of all patterns with regular constraints} by $Pat_{\Sigma,\mathcal{C}_{Reg}}$.
For some $(p,r_p)\in Pat_{\Sigma,\mathcal{C}_{Reg}}$ and $h\in H$, we say that $h$ is a \emph{$r_p$-valid substitution}
if $h(x) \in L(x)$ for all $x\in\var(p)$.
We extend the notion of pattern languages by the following. For any $(p,r_p)\in Pat_{\Sigma,\mathcal{C}_{Reg}}$ we denote by 
    $$L_E(p,r_p) := \{\ h(p)\ |\ h\in H, h(x)\in\Sigma^*\text{ for all } x\in\var(p), h \text{ is } r_p \text{-valid}\ \}$$
the \emph{erasing pattern language with regular constraints} of $(p,r_p)$ and by
    $$L_{NE}(p,r_p) := \{\ h(p)\ |\ h\in H, h(x)\in\Sigma^+\text{ for all } x\in\var(p), h \text{ is } r_p \text{-valid}\ \}$$
the \emph{non-erasing pattern language with regular constraints} of $(p,r_p)$.

\subsection{Nondeterministic 2-Counter Automata}

Usually, 2-counter automata are defined over input words utilising an input alphabet and the additional
use of two counters. In our setting, we consider a slight variation which assumes that the automaton always runs over an
empty input word.
A \emph{nondeterministic 2-counter automaton without input} (see e.g. \cite{Iberra1978}) is a 4-tuple $A = (Q,\delta,q_0,F)$ which consists of a set of states $Q$,
a transition function $\delta : Q \times \{0,1\}^2 \rightarrow \mathcal{P}(Q \times \{\-1,0,+1\}^2)$, an initial state $q_0\in Q$, and
a set of accepting states $F \subseteq Q$. A \emph{configuration} of $A$ is defined as a triple $(q,m_1,m_2)\in Q\times\N_0\times\N_0$
in which $q$ indicates the current state and $m_1$ and $m_2$ indicate the contents of the first and second counter.
We define the relation $\vdash_A$ on $Q\times\N_0\times\N_0$ by $\delta$ as follows. For two configurations
$(p,m_1,m_2)$ and $(q,n_1,n_2)$ we say that $(p,m_1,m_2) \vdash_A (q,n_1,n_2)$ if and only if there exist $c_1,c_2\in\{0,1\}$ and
$r_1,r_2\in\{-1,0,+1\}$ such that
\begin{enumerate}
	\item if $m_i = 0$ then $c_i = 0$, otherwise if $m_i > 0$, then $c_i = 1$, for $i\in\{1,2\}$,
	\item $n_i = m_i + r_i$ for $i\in\{1,2\}$,
	\item $(q,r_1,r_2)\in \delta(p,c_1,c_2)$, and
	\item we assume if $c_i = 0$ then $r_i \neq -1$ for $i\in\{1,2\}$.
\end{enumerate}
Essentially, the machine checks in every state whether the counters equal $0$ and then changes the value of each counter by at most one per transition
before entering a new state. A \emph{computation} is a sequence of configurations. An \emph{accepting computation} of $A$ is a
sequence $C_1,...,C_n\in (Q\times\N_0\times\N_0)^n$ with $C_1 = (q_0,0,0)$, $C_i \vdash_A C_{i+1}$ for all $i\in\{1,...,n-1\}$, and $C_n\in F\times\N_0\times\N_0$ for some $n\in\N$.

We \emph{encode} configurations of $A$ by assuming $Q = \{q_0,...,q_e\}$ for some $e\in\N_0$ and defining a function 
$enc$ $:$ $Q\times\N_0\times\N_0 \rightarrow \{0,\#\}^*$ by 
$$enc(q_i,m_1,m_2) := 0^{1+i}\#0^{5+2m_1}\#0^{5+2m_2}.$$
Notice that each state $q_i$ is mapped to a word $0^{i+1}$ and that each number $m_i$ is
mapped to an odd number $0^{5+2m_i}$ where $0^5$ denotes $0$, $0^7$ denotes $1$, $0^9$ denotes $2$ and so on.
This is extended to encodings of computations by defining for every $n\geq 1$ and every sequence $C_1,...,C_n\in Q\times\N_0\times\N_0$
$$ enc(C_1,...,C_n) := \#\#\ enc(C_1)\ \#\#\ ...\ \#\#\ enc(C_n)\ \#\#. $$
This encoding of configurations and computations is specifically chosen for its utility in proving Theorem \ref{theorem:pattern-regular-constraints-erasing-equiv-undecidable}.

Furthermore, define the set of accepting computations $$\mathtt{ValC}(A) := \{enc(C_1,...,C_n)\ |\ C_1,...,C_n \text{ is an accepting computation of } A\}$$
and let $\mathtt{InvalC}(A) = \{0,\#\}^*\setminus \mathtt{ValC}(A)$.
The emptiness problem for deterministic 2-counter-automata with input is undecidable (cf. e.g. \cite{Iberra1978,Minsky1961}), thus it is also undecidable
whether a nondeterministic 2-counter automaton without input has an accepting computation \cite{Freydenberger2010, Jiang1995}.

\subsection{Known Results}

The membership problems of both, the erasing and non-erasing pattern languages, have been shown to be NP-complete \cite{DBLP:journals/jcss/Angluin80,Jiang1994}.
Hence, we observe the following for patterns with regular constraints.

\begin{corollary}
	Let $(\alpha,r_\alpha) \in Pat_{\Sigma,\mathcal{C}_{Reg}}$ and $w\in\Sigma^*$.
	The decision problem of whether $w \in L_X(\alpha,r_\alpha)$ for $X \in \{E,NE\}$ is NP-complete.
\end{corollary}

Indeed, we immediately obtain NP-hardness in both cases by the previous results shown in \cite{DBLP:journals/jcss/Angluin80,Jiang1994} for patterns.
NP-containment follows by knowing that a valid certificate results in a substitution of $\alpha$ which has at most length $|w|$.

One other notable problem regarding patterns is the inclusion problem.
The undecidabilities of the inclusion problems for patterns in the erasing and non-erasing cases
have been initially shown by Jiang et al. \cite{Jiang1995} for unbounded alphabets and have been refined and extended to finite alphabets of sizes greater or equal to $2$ 
in \cite{Freydenberger2010, BREMER201215}. Hence we have the following.

\begin{theorem}\label{theorem:pattern-inclusion-undecidability}\cite{Jiang1995, Freydenberger2010, BREMER201215}
	Let $\alpha,\beta\in Pat_\Sigma$.
	In general, for all alphabets $\Sigma$ with $|\Sigma| \geq 2$, it is undecidable to answer whether
    \begin{enumerate}
        \item $L_E(\alpha) \subseteq L_E(\beta)$, or
        \item $L_{NE}(\alpha) \subseteq L_{NE}(\beta)$.
    \end{enumerate}
\end{theorem}

From that, we immediately obtain the following for patterns with regular constraints.

\begin{corollary}\label{corollary:pattern-regular-constraints-inclusion-undecidability}
	Let $(\alpha,r_\alpha),(\beta,r_\beta)\in Pat_{\Sigma,\mathcal{C}_{Reg}}$.
	In both, the terminal-free and the non terminal-free cases for $\alpha$ and $\beta$ we have in general,
	for all alphabets $\Sigma$ with $|\Sigma| \geq 2$, that it is undecidable to answer whether
	\begin{enumerate}
        \item $L_E(\alpha,r_\alpha) \subseteq L_E(\beta,r_\beta)$, or
        \item $L_{NE}(\alpha,r_\alpha) \subseteq L_{NE}(\beta,r_\beta)$.
    \end{enumerate}
\end{corollary}

Indeed, the general results follow immediately from Theorem \ref{theorem:pattern-inclusion-undecidability}.
Additionally, in the terminal-free cases, we can reduce the general versions to the terminal free versions by substituting
each terminal letter $\ta\in\Sigma$ which occurs in a pattern $\alpha$ by a new variable $x_\ta$
and setting $L(x_\ta) = \{\ta\}$. This results in effectively the same problem instances
without using terminals in the pattern words.

\section{Undecidability of E-Pattern Language Equivalence}
The main result of this paper considers the equivalence problem for erasing pattern languages with regular constraints.
In particular, we show that this problem is undecidable. 

\begin{theorem}\label{theorem:pattern-regular-constraints-erasing-equiv-undecidable}
	Let $(\alpha,r_\alpha),(\beta,r_\beta)\in Pat_{\Sigma,\mathcal{C}_{Reg}}$.
	In general, it is undecidable to decide whether $L_{E}(\alpha,r_\alpha) = L_{E}(\beta,r_\beta)$
	for all alphabets $\Sigma$ with $\Sigma \geq 2$.
\end{theorem}

The rest of this section is dedicated to show Theorem \ref{theorem:pattern-regular-constraints-erasing-equiv-undecidable}.
Roughly based on the idea of the proof of undecidability of the inclusion
problem for pattern languages in the case of finite alphabets, given by Freydenerger and Reidenbach \cite{Freydenberger2010},
we reduce the question whether some non-deterministic 2-counter automaton without input $A$ has some accepting computation
to the problem of whether the erasing pattern languages of two patterns with regular constraints are equal.
The first is known to be undecidable out of which the undecidability of the second problem follows.
In contrast to the proof given in \cite{Freydenberger2010}, the constructed patterns and predicates (to be explained later) had to 
be notably adapted to work for the case considered here.

Let $A =  (Q,\delta,q_0,F)$ be some non-deterministic 2-counter automaton without input.
We construct two patterns with regular constraints $(\alpha,r_\alpha),(\beta,r_\beta)\in Pat_{\Sigma,\mathcal{C}_{Reg}}$
such that $L_E(\alpha,r_\alpha) = L_E(\beta,r_\beta)$ if and only if $\mathtt{ValC}(A) = \emptyset$.

We start with the binary case and assume $\Sigma = \{0,\#\}$.
First, we construct $(\alpha,r_\alpha)$. We set the pattern $\alpha$ to
$$ \alpha = x_v\ \alpha_1\ x_v\ \tilde{y}$$
for variables $x_v, \tilde{y}, \alpha_1$. Let $v = 0\#^30$.
We then define the regular constraint $r_\alpha$ for $\alpha$ by
$L_E(x_v) := \{\varepsilon, v\}$,
$L_E(\alpha_1) := \{\ 0w0 \in \Sigma^*\ |\ w\in\Sigma^* \text{ and } |w|_{\#^3} = 0\} \cup \{\varepsilon\}$, and
$L_E(\tilde{y}) := \{\ w\in\Sigma^*\ |\ w \neq vuv \text{ for all } u\in\Sigma^* \text{ with } u\in L_E(\alpha_1)\setminus\{\varepsilon\}\ \}$.
Notice that the given regular constraints won't allow $\tilde{y}$ to be substituted to anything
we can obtain with $h(x_v\alpha_1 x_v)$ in the case of $x_v$ and $\alpha_1$ not being substituted by the empty word, but may be substituted to everything else.
Next, we construct $(\beta,r_{\beta})$. We set the pattern $\beta$ to
$$ \beta = \hat{\beta}_1\ ...\ \hat{\beta}_{\mu}\ \tilde{z} $$
such that $\tilde{z}$ is a new variable and 
$\hat{\beta}_1, ... ,\hat{\beta}_\mu$ are terminal free patterns defined by
$\hat{\beta}_i = x_i\ \upgamma_i\ x_i$ for new variables $x_i$ and some later specified terminal free pattern $\upgamma_i\in X^*$ for all $i\in[\mu]$.
We assume that each variable in $\var(\upgamma_i)$ only appears in $\upgamma_i$
and define $r_{\upgamma_i}$ as the set of regular constraints on the variables occurring in $\upgamma_i$.
By the construction that follows we assume that for all $x\in\var(\upgamma_i)$ we always have $\varepsilon \in L(x)$.
Notice that each $x_i$ occurs $2$ times in $\beta$ for all $i\in[\mu]$ and also notice that for all $x\in\var(\beta)$ we have
that $\varepsilon\in L(x)$.
We define the regular constraints $r_\beta$ on $\beta$ by setting
$L(x_i) := \{\varepsilon, v\}$ for all $i\in[\mu]$ and
$L(\tilde{z}) := L(\tilde{y})$.
Additionally we add all regular constraints defined by $r_{\upgamma_i}$ to $r_\beta$ for all $i\in[\mu]$.
Further, we from now on assume that for all $w\in L_E(\upgamma_i, r_{\upgamma_i})$ we have either $w = \varepsilon$ or
$w = 0u0$ for $u\in\Sigma^*$ with $|u|_{\#^3} = 0$.
This assumption holds by the construction that follows.

Using the construction up to this point and the assumptions we made so far, we first show
the following property.

\begin{lemma}\label{lemma:all-words-beta-in-alpha}
    We have $L_E(\beta,r_\beta) \subseteq L_E(\alpha,r_\alpha)$.
\end{lemma}
\begin{proof}
    Let $w\in L_E(\beta, r_\beta)$.
    Then, there exists some $r_\beta$-valid $h\in H$ such that $h(\beta) = w$.
    We differentiate between two main cases.

    For the first case, assume $h(\beta) = h(\hat{\beta}_1 ... \hat{\beta}_\mu\tilde{z}) = v0u0v$ for some $u\in\Sigma^*$ with $|u|_{\#^3} = 0$.
    By that, we know that $h(\beta) \notin L(\tilde{y})$ as $0u0\in L(\alpha_1)\setminus\{\varepsilon\}$.
    Let $h'\in H$ be some substitution.
    Set $h'(x_v) = v$, $h'(\alpha_1) = 0u0$, and $h'(\tilde{y}) = \varepsilon$.
    We have that $h'$ is $r_\alpha$-valid.
    We get $h'(\alpha) = h'(x_v\alpha x_v\tilde{y}) = v0u0v$.
    So, $h(\beta) = h'(\alpha)$ and by that $h(\beta)\in L_E(\alpha,r_\alpha)$.

    In the second case, assume $h(\beta) \neq v0u0v$ for any $u\in\Sigma^*$ with $|u|_{\#^3} = 0$.
    Then $h(\beta) \in L(\tilde{y})$.
    Let $h'\in H$ such that $h'(\tilde{y}) = h(\beta)$ and $h'(x_v) = h'(\alpha_1) = \varepsilon$.
    Then $h'$ is $r_\alpha$-valid and we get $h'(\alpha) = h(\beta)$. By that, $h(\beta)\in L_E(\alpha,r_\alpha)$ which concludes this lemma.
    \qed
\end{proof}

So by now we know that all words in the language of the pattern $(\beta,r_\beta)$ are also in the language generated by the
pattern $(\alpha,r_\alpha)$. Next, we show the rather immediate result that all words in the language of $(\alpha,r_\alpha)$ 
that do not follow a specific form are also in the language generated by $(\beta,r_\beta)$. This fact is important for the construction
that follows.

\begin{lemma}\label{lemma:trivial-words-alpha-in-beta}
    Let $h\in H$ such that $h(\alpha)\in L_E(\alpha,r_\alpha)$.
    If $h(\alpha) \neq v0u0v$ for all $u\in\Sigma^*$ with $|u|_{\#^3} = 0$, 
    then $h(\alpha) \in L_E(\beta,r_\beta)$.
\end{lemma}
\begin{proof}
    Select $h'\in H$ such that $h'(\tilde{z}) = h(\alpha)$ and $h'(x) = \varepsilon$ for all other $x\in\var(\beta)$ with $x\neq\tilde{z}$.
    By assumption, we know that $h(\alpha) \in L(\tilde{z})$ and that $\varepsilon \in L(x)$ for all other $x\in\var(\beta)$ with $x\neq\tilde{z}$.
    Hence, $h'$ is $r_\beta$-valid. We get $h'(\beta) = h(\alpha)$, thus we have $h(\alpha)\in L_E(\beta,r_\beta)$. This concludes this lemma.
    \qed
\end{proof}

Finally, we show that substitutions of $(\alpha,r_\alpha)$ that follow that specific form can only
be obtained from $(\beta,r_\alpha)$ if and only if there exists some $\hat{\beta}_i$ for which we
have $h(\alpha) = h'(\hat{\beta_i})$.

\begin{lemma}\label{lemma:predicate-necesssary-substitutions}
    Let $h\in H$ such that $h(\alpha) \in L_E(\alpha,r_\alpha)$.
    If $h(\alpha) = v0u0v$
    for some $u\in\Sigma^*$ with $|u|_{\#^3} = 0$,
    then $h(\alpha) \in L_E(\beta,r_\beta)$ if and only if there exists some $i\in[\mu]$
    and $r_\beta$-valid $h'\in H$ with $h'(\hat{\beta}_i) = h(\alpha)$.
\end{lemma}
\begin{proof}
    Let $h\in H$ be given as in the claim.
    So, we have $h(\alpha) = h(x_v\alpha_1x_v\tilde{y}) = v0u0v$ for some $u\in\Sigma^*$ with $|u|_{\#^3} = 0$.
    For the first direction assume $h(\alpha) \in L_E(\beta,r_\beta)$.
    Then, there exists some $r_\beta$-valid substitution $h'\in H$ such that $h'(\beta) = h(\alpha) = v0u0v$.
    By construction, we know that $h(\alpha) \notin L(\tilde{z})$.
    Also, we know that for all $i\in[\mu]$ and for all $w\in L_E(\upgamma_i,r_{\upgamma_i})$ we either have $w = \varepsilon$
    or $w = 0u0$ for some $u\in\Sigma^*$ with $|u|_{\#^3} = 0$.
    Hence, we have $\#^3\notin\Fact(h'(\upgamma_1...\upgamma_\mu))$.
    So, there exists some $i\in[\mu]$ such that $h'(x_i) \neq \varepsilon$ and by that $h'(x_i) = v$.
    As $|v0u0v|_{\#^3} = 2$ and $|\beta|_{x_i} = 2$, we immediately get that $h'(x_i\upgamma_ix_i) = v h'(\upgamma_i) v = v0u0v$
    and by that $h'(\upgamma_i) = w'$. In particular, for all other $x\in\var(\beta)$ with $x\neq x_i$ and $x\notin\var(\upgamma_i)$
    we have $h'(x) = \varepsilon$. This concludes this direction.
    The other direction immediately follows by the assumption and by setting all variables $x\in\var(\beta)$ with $x\notin\var(\hat{\beta}_i)$
    to $x = \varepsilon$. Then $h'(\hat{\beta}_i) = h(\alpha)$ and by that $h(\alpha)\in L_E(\beta,r_{\beta})$.
    \qed
\end{proof}

Now, we know that if $h(\alpha)$ has some precise form, that $h(\alpha) \in L_E(\beta,r_\beta)$
if and only if there exists some $\hat{\beta}_i$ which we can use to obtain that specific $h(\alpha)$.
By that, we obtain the following for all words which are not in both languages.

\begin{corollary}
    For some $r_\alpha$-valid substitution $h\in H$ we have $h(\alpha) \notin L_E(\beta,r_\beta)$
    if and only if $h(\alpha) = v\ w'\ v$ for some $w'\in\Sigma^*$ with 
    $w' \in \{\ 0u0\ |\ u\in\Sigma^*, |u|_{\#^3} = 0\}$
    and for all $i\in[\mu]$ we have
    $w' \notin L_E(\upgamma_i,r_{\upgamma_i})$.
\end{corollary}
\begin{proof}
    Immediately follows by Lemma \ref{lemma:predicate-necesssary-substitutions} and Lemma \ref{lemma:trivial-words-alpha-in-beta}
    and the fact that all other words of $L_E(\alpha,r_\alpha)$ are contained in $L_E(\beta,r_\beta)$.
    \qed
\end{proof}

We say that a word $w\in\Sigma^*$ is of \emph{good structure} or a \emph{computation} if
$w \in L_G$ with $L_G = ((\#\#00^*\#0^5(00)^*\#0^5(00)^*)^+\#\#)$. Otherwise, we say that $w$ is of \emph{bad structure}.
Clearly, all encodings of computations of $A$ are words of good structure.

From now on, let $h\in H$ be some $r_\alpha$-valid substitution and assume $h(\alpha) = v0u0v$ for some 
$u\in\Sigma^*$ with $|u|_{\#^3} = 0$ and $v = 0\#^30$ as before.

We now have to construct $\hat{\beta}_1$ to $\hat{\beta}_\mu$ such that if
$u$ is not an encoding of a valid computation of $A$, then we have that there exists some $i\in[\mu]$
with $\hat{\beta}_i = x_i \upgamma_i x_i$ and $w \in L_E(\upgamma_i,r_{\upgamma_i})$.
Once we have that, we know that for any $r_\alpha$-valid $h'\in H$ we have $h'(\alpha)\notin L_E(\beta,r_\beta)$
if and only if $h'(\alpha) = v0w_c0v$ for any $w_c\in\mathtt{ValC}(A)$,
concluding this reduction. 

For all $i\in[\mu]$ we call the pattern with regular constraints $(\upgamma_i,r_{\upgamma_i})$ a predicate.
We construct each predicate independently, hence we omit their specific indexes from now on.
Assume each predicate does not share its index with any other predicate and assume the total number of predicates to be $\mu\in\N$.
As we will see, the total number of predicates is bound by the number of non-final states $|Q \setminus F|$,
the number of invalid transitions not found in $\delta$, and a constant number of predicates considering the basic structure of
encodings of computations.
Notice, that each constructed predicate ensures that for all $r_\beta$-valid $h'\in H$ we have
$h'(\upgamma) = \varepsilon$ or $h'(\upgamma) = 0u'0$ for some $u'\in\Sigma^*$ with $|u'|_{\#^3} = 0$,
which satisfies our initial assumption.

(1) First, we construct a predicate which can be used to obtain all substitutions in which $h(u)$ is not of good structure
and which does not start with an encoding of the initial configuration $(q_0,0,0)$. For that, let $\upgamma = y$
for a new and independent variable $y\in X$ and set 
$$L(y) := \{\varepsilon\} \cup \{\ 0u'0\ |\ u'\in\Sigma^*, |u'|_{\#^3} = 0, u'\in L_{gs}\}$$
for
$$L_{gs} := \overline{L(\ \#\#0\#0^5\#0^5\ (\#\#0^+\#0^5(00)^*\#0^5(00)^*)^*\ \#\#\ )}.$$
Then, if $u$ is not of good structure or does not start with a valid encoding of the initial configuration, we can 
define a $r_\beta$-valid $h'\in H$ such that $h'(\upgamma) = 0u0$.

(2) Next, we construct predicates which can be used to obtain all substitutions which end in an encoding of a configuration
that is not in a final state. So, for all $q_j\in Q\setminus F$ we define a new and independent predicate $\upgamma = y$
for respectively new and independent variables $y \in X$ such that 
$$ L(y) := \{\varepsilon\} \cup \{\ 0u'0\ |\ u'\in\Sigma^*, |u'|_{\#^3} = 0, u'\in \Sigma^*\cdot L(\#\#0^{1+j}\#0^+\#0^+\#\#)\}. $$
Then, if $u$ ends in an encoding of a configuration of $A$ which contains no final state, we can obtain
a $r_\beta$-valid substitution $h'\in H$ such that $h'(\upgamma) = 0u0$.

(3) Now, we have to make sure that in a single step the value of no counter is changed by more than one.
For that, we construct four predicates, each corresponding to the value of either the first or second counter being either increased or decreased
by more than one (in a single step of an encoding of a computation).
First, we construct a new and independent predicate $\upgamma$ which can be used if the first counter is increased by more than one
in a single step. Let $\upgamma = y_1\ x_1\ y_2\ x_1\ y_3$ for new and independent variables $y_1,y_2,y_3,x_1\in X$
and set
\begin{align*}
    L(y_1) &:= \{\varepsilon\} \cup \{\ 0u0\#0\ |\ u\in\Sigma^*, |u|_{\#^3} = 0\}, \\
    L(y_2) &:= \{\varepsilon\} \cup L(0^4\#0^50^*\#\#0^+\#0^4\mathbf{00(00)^+}), \\
    L(y_3) &:= \{\varepsilon\} \cup \{\ 0\#0u0\ |\ u\in\Sigma^*, |u|_{\#^3} = 0\}, \\
    L(x_1) &:= \{00\}^*.
\end{align*}
Then, if $h(u)$ has a factor $\#0^50^m\#0^50^n\#\#0^{1+j}\#0^50^m\mathbf{00(00)^k}\#$ for $m,n,j,k\in\N$ and $k\geq 2$, which corresponds to a part of an encoding of the
first counter being increased by more than one (see bold numbers), we can find a $r_\beta$-valid substitution $h'\in H$ for which we have $h'(\upgamma) = 0u0$.
All other words obtainable from $\upgamma$ are words of bad structure, i.e.,
they are not in $L_G$ if any of the variables $y_1$, $y_2$, or $y_3$ is substituted by the empty word as
$L(y_1)L(x_1)L(x_1) \cap L_G = \emptyset$, $L(x_1)L(y_2)L(x_1) \cap L_G = \emptyset$, $L(x_1)L(x_1)L(y_3) \cap L_G = \emptyset$,
$L(y_1)L(x_1)L(y_2)L(x_1) \cap L_G = \emptyset$, $L(y_1)L(x_1)L(x_1)L(y_3) \cap L_G = \emptyset$, and $L(x_1)L(y_2)L(x_1)L(y_3) \cap L_G = \emptyset$.
Also, we cannot get $|h'(\upgamma)|_{\#^3} > 0$.
The cases of the first counter being decreased by more than one, the second counter being increased by more than one, and the second counter being
decreased by more than one can all the constructed in an analogue manner, hence we omit their specific constructions here.
They only differ in their definition of $L(y_2)$, in particular the placement of either the border $\#\#$ or the position of $00(00)^+$.

By now, only if $u$ corresponds to a word of good structure in which every subsequent pair of encodings of configurations
in which either no counter, one counter, or both counters are increased or decreased by at most one, we cannot find a predicate $\upgamma$ and
a $r_\beta$-valid substitution $h'\in H$ such that $h'(\upgamma) = u$.
That already contains all encodings of valid computations of $A$, however we may still get encodings of computations in which
two subsequent configurations do not correspond to any valid transition. 

(4) So, in a last step, we construct predicates for each
invalid pair of consecutive configurations based on the definition $\delta$ in $A$. 
For all $q_k,q_j\in Q$, $c_1,c_2\in\{0,1\}$, and $r_1,r_2\in\{-1,0,1\}$ with
$(q_k,r_1,r_2)\notin\delta(q_j,c_1,c_2)$ we define a new and independent predicate $\upgamma$ which can be used to obtain encodings of computations in which 
such an (invalid) transition is used. We demonstrate the construction using an examplary case by setting $c_1 = 1$, $c_2 = 1$, $r_1 = +1$,
and $r_2 = 0$. Let 
$$\upgamma = y_1\ x_1\ y_2\ x_2\ y_3\ x_1\ y_4\ x_2\ y_5$$ 
for new and independent variables $y_1,...,y_5,x_1,x_2\in X$ and set
\begin{align*}
    L(y_1) &:= \{\varepsilon\} \cup \{\ u'0\#\#0^{1+j}\#0\ |\ u'=\varepsilon \text{ or } u'=0u'', |u''|_{\#^3}=0, u',u''\in\Sigma^*\ \}, \\
    L(y_2) &:= \{\varepsilon\} \cup \{0^6\#0\}, \\
    L(y_3) &:= \{\varepsilon\} \cup \{0^6\#\#0^{1+i}\#0^6\mathbf{00}\}, \\
    L(y_4) &:= \{\varepsilon\} \cup \{0\#0^4\}, \\
    L(y_5) &:= \{\varepsilon\} \cup \{\ 0^3\#\#0u'\ |\ u'=\varepsilon \text{ or } u'=u''0, |u''|_{\#^3} = 0, u',u''\in\Sigma^*\}, \\
    L(x_1) &:= \{00\}^*, \text{ and} \\
    L(x_2) &:= \{00\}^*.
\end{align*}
Then, if $h(u)$ contains a factor $$\#\#0^{j+1}\#0^{5}0^{2+2m_1}\#0^{5}0^{2+2m_2}\#\#0^{i+1}\#0^50^{2+2m_1}\mathbf{0^2}\#0^50^{2+2m_2}\#\#,$$
which corresponds to $(q_i,r_1,r_2)\notin\delta(q_j,c_1,c_2)$, this predicate can be used to find a $r_\beta$-valid substitution $h'\in H$
for which we have $h'(\upgamma) = 0u0 = w$. Notice that each counter starts with a value $0^70^{2m_i}$ for $i\in[2]$ and $m_i\in\N_0$
instead of $0^50^{2m_i}$ as we assume both counters not to be zero in this example (by $c_1 = c_2 = 1$).
Predicates for all other cases can be constructed analogously by either switching the position of additional $\mathbf{0}'s$ (marked with bold letters in the construction), or removing
one or both occurrences of either $x_1$ or $x_2$ (and reducing the number of 0's in the corresponding part by 2) if $c_1 = 0$ or $c_2 = 0$ respectively.

Whats left to make sure is that for all $r_\beta$-valid $h'\in H$ we have that $h'(\upgamma) = 0u0$ for $u\in\Sigma^*$ such that
$|u|_{\#^3} = 0$ and $u\notin\mathtt{ValC}(A)$, even if some variables in $\upgamma$ are substituted with the empty word.
First, by the way we defined $L(y_i)$ and $L(x_j)$ for $i\in\{1,...,5\}$ and $j\in\{1,2\}$, we cannot obtain words in which $\#^3$ occurs as a factor.
Second, notice that substitutions that only map $y_2$,$y_3$, or $y_4$ to nonempty words, directly result in words of bad structure due to their suffixes and prefixes.
If we only have $h'(y_1) \neq \varepsilon$ or $h'(y_5) \neq \varepsilon$, then either the suffix or the prefix respectively results in bad structure.
The only potentially problematic substitution is if either all variables except $y_1$, $y_2$, $y_5$, and potentially occurrences of $x_1$ and $x_2$, 
or all variables except $y_1$, $y_4$, $y_5$, and potentially occurrences of $x_1$ and $x_2$ are substituted by the empty word.
Then we get a structure which resembles only one configuration. But then, we notice that either the first or the second counter
always has an even number of $0's$. This is not a valid encoding of a configuration, i.e., it is a word of bad structure. 
Hence, we cannot obtain $h'(\upgamma) = 0u'0$ with $u'\in\mathtt{ValC}(A)$.

Using all predicates, given some $r_\alpha$-valid $h\in H$, we can conclude that $h(\alpha)\notin L_E(\beta,r_\beta)$
if and only if $h(\alpha) = v0u0v$ such that $u\in\mathtt{ValvC}(A)$.
This decides the problem of whether $A$ has some accepting computation, hence the erasing equivalence problem for
pattern languages with regular constraints is undecidable in the binary case.
For larger alphabets, we may always restrict the alphabets used in the languages of the variables to the binary case, which allows for an reduction from the
binary case to all larger alphabet sizes. 
This concludes the proof of Theorem \ref{theorem:pattern-regular-constraints-erasing-equiv-undecidable}.

\section{Further Discussion}
As the constructed patterns in the previous reduction are both terminal-free, we have
immediately covered the general and terminal-free case together, as the latter can
be easily reduced to the first. We mention the following fact which formalizes the first
statement.

\begin{corollary}
    Let $(\alpha,r_\alpha),(\beta,r_\beta)\in Pat_{\Sigma,\mathcal{C}_{Reg}}$ such that $\alpha,\beta\in X^*$, i.e. $\alpha$
    and $\beta$ are terminal-free patterns.
    In general, it is undecidable to decide whether $L_{E}(\alpha,r_\alpha) = L_{E}(\beta,r_\beta)$
    for all alphabets $\Sigma$ with $\Sigma \geq 2$.
\end{corollary}

With that, we obtain undecidability for nearly all problems regarding pattern languages with regular constraints.
The only open case is the equivalence problem of non-erasing pattern languages with regular constraints.
Using regular constraints, the problem becomes at least as hard as deciding the equivalence of two given regular languages
witnessed by the following example.

\begin{example}
    Let $(\alpha,r_\alpha),(\beta,r_\beta)\in Pat_{\Sigma,\mathcal{C}_{Reg}}$
    such that $\alpha = x$ and $\beta = y$ for some $x,y\in X$.
    Then $L_{NE}(\alpha,r_\alpha) = L_{NE}(\beta,r_\beta)$ if and only if 
    $L(x)\setminus\{\varepsilon\} = L(y)\setminus\{\varepsilon\}$.
\end{example}

Despite the most prominent open problem for patterns being undecidable in the case of pattern languages with regular constraints,
we see that even this problem, which is trivially decidable for patterns without regular constraints, becomes much harder in this setting. 
We propose the following open
question to which we have no definite conjecture so far.
An overview of the current state of patterns with regular constraints can be found in Table \ref{tab:pat-reg-const-state}.

\begin{question}
    Given $(\alpha,r_\alpha),(\beta,r_\beta)\in Pat_{\Sigma,\mathcal{C}_{Reg}}$,
	is it generally decidable to answer whether $L_{NE}(\alpha,r_\alpha) = L_{NE}(\beta,r_\beta)$?
\end{question}

\begin{table}[H]
    \label{tab:pat-reg-const-state}
    \centering
    \setlength\tabcolsep{0.4cm}
    {\renewcommand{\arraystretch}{1.2}
        \begin{tabular}{ | l | c | c | }
            \hline
            \mbox{Problem} & General & Terminal-Free \\ 
            \hline \hline 
            E-Membership & NP-complete & NP-complete \\  
            \hline 
            E-Inclusion & Undecidable & Undecidable \\
            \hline 
            E-Equivalence & Undecidable & Undecidable \\
            \hline 
            NE-Membership & NP-complete & NP-complete \\
            \hline 
            NE-Inclusion &  Undecidable & Undecidable \\
            \hline 
            NE-Equivalence & Open & Open \\
            \hline 
        \end{tabular}
    }
    \vspace{3mm}
    \caption{Current state regarding pattern languages with regular constraints}
\end{table}

%
%
\newpage
\bibliographystyle{splncs04}
\bibliography{patlang-regconst}

\end{document}